\documentclass[12pt,a4paper]{ouparticle}

\usepackage{amsmath}
\usepackage{amsfonts}
\usepackage{dsfont}
\usepackage{amssymb}
\usepackage{amsthm}

\usepackage{hyperref}
\usepackage{cite}

\begin{document}

\newtheorem{thm}{\textsl{Theorem}}[section]
\newtheorem{dfn}[thm]{\textsl{Definition}}
\newtheorem{prp}[thm]{\textsl{Proposition}}
\newtheorem{cor}[thm]{\textsl{Corollary}}
\newtheorem{lem}[thm]{\textsl{Lemma}}
\newtheorem{rmk}[thm]{\textsl{Remark}}

\numberwithin{equation}{section}

\title{\textbf{Extended Gini-type measures of risk \break{ and variability}}}

\author{%
\name{Mohammed Berkhouch\footnote{\textit{Corresponding author. Phone number: (+212)6 10 88 46 86}.}}
\address{\textit{LGII, ENSA, Ibn Zohr University, Agadir, Morocco.}}
\email{\texttt{E-mail: mohammed.berkhouch@edu.uiz.ac.ma}}
\and
\name{Ghizlane Lakhnati}
\address{\textit{LGII, ENSA, Ibn Zohr University, Agadir, Morocco.}}
\email{\texttt{E-mail: g.lakhnati@uiz.ac.ma}}
\and
\name{Marcelo Brutti Righi}
\address{\textit{Federal University of Rio Grande do Sul, Porto Alegre, Brazil.}}
\email{\texttt{E-mail: marcelo.righi@ufrgs.br}}}

\abstract{The aim of this paper is to introduce a risk measure that extends the Gini-type measures of risk and variability, the Extended Gini Shortfall, by taking risk aversion into consideration. Our risk measure is coherent and catches variability, an important concept for risk management. The analysis is made under the Choquet integral representations framework. We expose results for analytic computation under well-known distribution functions. Furthermore, we provide a practical application.
\begin{flushleft}
\end{flushleft}
}

\date{\today}


\maketitle

\section{Introduction}
\label{sec1}

In modern risk management, a large number of risk measures have been proposed in the literature. These measures are mappings from a set of random variables (financial losses) to real numbers. At first, the focus were on the variability over an expected return, as is the case for the well-known variance. After the collapses and crises in financial systems, a prominent trend associated with tail-based risk measures has emerged, especially with the most popular ones nowadays: the Value-at-Risk (VaR) and the Expected Shortfall (ES). However, this kind of risk measures does not capture the variability of a financial position, a primitive but relevant concept. In order to solution this issue, some authors have proposed and studied specific examples of risk measures.

In this sense, Fischer \cite{8} considered combining the mean
and semi-deviations. Regarding tail risk, Furman and Landsman \cite{10} proposed a measure that weighs the mean and standard deviation in the truncated tail by VaR, while Righi and Ceretta \cite{17} considered penalizing the ES by
the dispersion of losses exceeding it. From a practical perspective, Righi and
Borenstein \cite{18} explored this concept, calling the approach as loss-deviation, for portfolio optimization. In a more general fashion, Righi \cite{19} presents results and examples about compositions of risk and variability measures in order to ensure solid theoretical properties.

Recently, Furman et al. \cite{11} introduced the Gini Shortfall (GS) risk measure which is coherent and satisfies co-monotonic additivity. GS is a composition between ES and tail based Gini coefficient. However, GS supposes that all individuals have the same attitude towards risk, while agents differ in the way they take personal decisions that involve risk because of discrepancies in their risk aversion. To incorporate such psychological behavior in tail risk analysis, we introduce a generalized version of the GS. This risk measure, called Extended Gini Shortfall (EGS), captures the notion of variability, satisfies the co-monotonic additivity property, and it is coherent under a necessary and sufficient condition for its loading parameter. The consideration of the decision-maker risk aversion, joined to these properties, is in consonance to what agents seek when searching for a suitable measure of risk. The approach followed in this article leads us to a new family of spectral risk measures, proposed by Acerbi \cite{1}, with an attractive weighting function.

In this sense, we discuss, in a separated manner, the properties from the variability term and our composed risk measure. Moreover, we discuss in details the role of each parameter in the mentioned weighting function. Furthermore, we expose results on analytic formulations for computation of EGS under known distribution functions. Our focus in this paper is on theoretical results, but this approach gives rise to further forthcoming investigations. In this sense, risk forecasting of our new family of risk measures is a subject that deserves a further and separate survey which will be made in a forthcoming work.

The rest of the paper is organized as follows. In Section 2, we present and discuss some preliminaries such as essential properties of measures of risk and variability, as well we elucidate the role of the signed Choquet integral. In section 3, we start with the Classical and Extended Gini functionals in order to introduce the concept and explore properties about what we call Tail Extended Gini functional and Extended Gini Shortfall. In Section 4, we give the closed-form of our risk measures class for elliptical distributions and then derive the uniform, Normal and Student-t cases. Section 5 illustrates an application of the introduced risk measures class in practice.

\section{Preliminaries}
\label{sec2}

We first introduce some basic notation. Let $(\Omega,\mathcal{A},\mathbb{P})$ be an atomless probability space. All equations and inequalities are in the $\mathbb{P}$ almost surely sense. Let $L^{q}$, $q\in [0,\infty)$, denote the set of all random variables (rv's from now on) in $(\Omega,\mathcal{A},\mathbb{P})$ with finite q-th moment and $L^{\infty}$ be the set of all essentially bounded rv's. Throughout this paper, $X \in L^{0}$ is a rv modeling financial losses (profits) when it has positive (negative) values. For every $X \in L^{0},\, F_{X}$ denotes the cdf of $X$, and $U_{X}$ denote any uniform $[0,1]$ rv such that the equation $F^{-1}_{X}(U_{X})=X$ holds. The existence of such rv's is assured in
 R\"{u}schendorf (\cite{21}, Proposition 1.3). We denote $x_p$ as the p-quantile of $X$. Two rv's $X$ and $Y$ are co-monotonic when
$(X(\omega)-X(\omega'))(Y(\omega)-Y(\omega'))\geq 0 \text{ for } (\omega,\omega')\in \Omega\times\Omega \quad (\mathbb{P}\times\mathbb{P})\text{-almost surely}$. Throughout the present paper, we deal with several convex cones $\mathcal{X}$ of rv's, of which $\mathcal{X}= L^{1}$ is of particular importance and $L^{\infty}$ is always contained in $\mathcal{X}$.

We begin by exposing definitions of both risk and variability measures. We assume throughout the paper that all functionals respect the following property, which is essential in order to obtain a functional directly from its distribution function.

\begin{dfn}
A functional $f$ is said to be law invariant if it fulfills the following property:

(A) \textit{Law Invariance:} If $X\in \mathcal{X}$ and $Y\in \mathcal{X}$ have the same distributions under $\mathbb{P}$, succinctly $X\stackrel{d}{=}Y$, then $f(X)=f(Y)$.
\end{dfn}

\begin{dfn}
A risk measure is a functional $\rho:\mathcal{X}\rightarrow(-\infty,\infty]$, which may fulfills the following properties:\\

(B1) \textit{Monotonicity:} $\rho(X)\leq \rho(Y)$ when $X,Y \in \mathcal{X}$ are such that $X\leq Y$ $\mathbb{P}$-almost surely.\\

(B2) \textit{Translation invariance:} $\rho (X+m)=\rho(X)+m$ for all $m\in \mathbb{R}$ and $X\in \mathcal{X}$.\\

(A1) \textit{Positive homogeneity:} $\rho (\lambda X)=\lambda \rho(X)$ for all $\lambda >0$ and $X\in \mathcal{X}$.\\

(A2) \textit{Subadditivity:} $\rho(X+Y)\leq \rho(X)+\rho(Y)$ for all $X,Y\in \mathcal{X}$.\\

(A3) \textit{Co-monotonic additivity:} $\rho(X+Y)=\rho(X)+\rho(Y)$ for every co-monotonic pair $X,Y\in\mathcal{X}$.\\

A risk measure is \textit{monetary} if it satisfies properties (B1) and (B2), and it is \textit{coherent} if it satisfies furthermore (A1) and (A2).
\end{dfn}

\begin{rmk}
For interpretations of these properties, we refer the reader to F\"{o}llmer and Schied (\cite{9}, Chap 4), Delbaen \cite{6}, and McNeil et al. \cite{16}. For example, both functionals VaR and ES are monetary and co-monotonically additive, whereas ES is even coherent.
\end{rmk}

\begin{dfn}
A functional $\nu:\mathcal{X}\rightarrow[0,\infty]$ is a measure of variability, which may fulfills the following properties\footnote{Inspired from the \textit{deviation measures} notion of Rockafellar et al. \cite{20}.}:\\

(C1) \textit{Standardization:} $\nu(c)=0$ for all $c\in \mathbb{R}$.\\

(C2) \textit{Location invariance:} $\nu(X+c)=\nu(X)$ for all $c\in \mathbb{R}$ and $X\in \mathcal{X}$.\\

(A1) \textit{Positive homogeneity:} $\nu (\lambda X)=\lambda \nu(X)$ for all $\lambda >0$ and $X\in \mathcal{X}$.\\

(A2) \textit{Subadditivity:} $\nu(X+Y)\leq \nu(X)+\nu(Y)$ for all $X,Y\in \mathcal{X}$\\

(A3) \textit{Co-monotonic additivity:} $\nu(X+Y)=\nu(X)+\nu(Y)$ for every co-monotonic pair $X,Y\in\mathcal{X}$.\\

A measure of variability is coherent if it further satisfies (C1), (C2), (A1) and (A2).
\end{dfn}

\begin{rmk}
 For instance, the most classical measures of variability are the Variance and the Standard Deviation. The variance functional satisfies properties (A), (C1), (C2) but not (A1) or (A2), hence it is not coherent. On the other side, the standard deviation functional, since satisfying all aforementioned properties, is coherent. Neither the variance nor the standard deviation is co-monotonically additive.
\end{rmk}

The notion of signed Choquet integral plays a pivotal role thereafter. It originates from Choquet \cite{4}, in the framework of capacities, and is further characterized and studied in decision theory by Schmeidler (\cite{22}, \cite{23}).

\begin{dfn} A function $h:[0,1]\rightarrow \mathbb{R}$ is called a distortion function when it is non-decreasing and satisfies the boundary conditions $h(0)=0$ and $h(1)=1$. Let $h$ be a distortion function, the functional defined by the equation:
\begin{equation}\label{1}
  I(X)= \int_{0}^{\infty}(1-h(F_{X}(x)))dx-\int_{-\infty}^{0}h(F_{X}(x))dx
\end{equation}
for all $X\in\mathcal{X}$ is called the (\textit{increasing}) \textit{Choquet integral}. Whenever $h:[0,1]\rightarrow \mathbb{R}$ is of finite variation, $I$ is called the \textit{signed Choquet integral}.
\end{dfn}

\begin{rmk}
When $h$ is right-continuous, then equation(\ref{1}) can be rewritten as (Wang et al. \cite{25}):
\begin{equation}\label{2}
  I(X)= \int_{0}^{1}F^{-1}_{X}(t)dh(t).
\end{equation}
Furthermore, when $h$ is absolutely continuous, with $\phi$ a function such that $dh(t)=\phi(t)dt$, then equation(\ref{2}) becomes:
\begin{equation}
  I(X)= \int_{0}^{1}F^{-1}_{X}(t)\phi(t)dt.
\end{equation}
In this case, $\phi$ is called the \textit{weighting functional} of the signed Choquet integral $I$.
\end{rmk}

\begin{rmk}
The signed Choquet integral is co-monotonically additive, as we can readily see from representation (\ref{2}) (cf. Schmeidler \cite{22}). Moreover we know from Yaari \cite{26} and F\"{o}llmer and Schied (\cite{9}, Theorem 4.88), that any law-invariant risk measure is co-monotonically additive and monetary if and only if it can be represented as a Choquet integral. Furthermore, the functional $I$ is sub-additive if and only if the function $h$ is convex (cf. Yaari \cite{26} and Acerbi \cite{1}). Moreover, as proved in Furman et al. \cite{11}, regarding the weighting functional $\phi$, the integral is monotone if and only if $\phi\geq 0$ on $[0,1]$ and it is sub-additive if and only if $\phi$ is non-decreasing on $[0,1]$. The major difference between a (an increasing) Choquet integral and a signed one is that the latter, being more general, is not necessarily monotone.
\end{rmk}

One of the practical and theoretical reasons for what we are particularly interested in signed Choquet integral is that we know that a suitable risk measure should be monotone as argued by Artzner et al. \cite{2}, but this issue is irrelevant for a measure of variability. In other words, signed Choquet integral is relevant as long as a measure of variability is concerned. The following theorem is enunciated with a complete proof in Furman et al. \cite{11}, it gives the characterization for co-monotonically additive and coherent measures of variability.

\begin{thm} \label{3} Let $\nu : L^{q}\rightarrow \mathbb{R}$ be any $L^{q}$-continuous functional. The following three statements are equivalent:\\

(i) $\nu$ is a co-monotonically additive and coherent measure of variability.\\

(ii) There is a convex function $h : [0,1]\rightarrow \mathbb{R}, h(0)=h(1)=0$, such that
\begin{equation}\label{6}
  \nu (X)= \int_{0}^{1}F^{-1}_{X}(u)dh(u),\quad X\in L^{q}.
\end{equation}

(iii) There is a non-decreasing function $g: [0,1]\rightarrow \mathbb{R}$ such that
\begin{equation}\label{5}
  \nu (X)= Cov[X, g(U_{X})],\quad X\in L^{q}.
\end{equation}
\end{thm}

Next, we recall a few partial orders of variability that have been popular in economics, insurance, finance and probability theory:

\begin{dfn} For $X,Y\in L^{1}$, $X$ is \textit{second-order stochastically dominated (SSD}) by $Y$, succinctly $X\prec _{SSD}Y$, if $\mathbb{E}[f(X)]\leq \mathbb{E}[f(Y)]$ for all increasing convex functions $f$.\\
If in addition, $\mathbb{E}[X]=\mathbb{E}[Y]$, then we say that $X$ is \textit{smaller than $Y$ in convex order}, succinctly $X\prec_{CX}Y$.\footnote{ We say equally $Y$ is a Mean Preserving Spread of $X$, succinctly $Y \,MPS\, X$.}\\
Under this framework we have the following properties for risk and variability measures:

(B3) \textit{SSD-monotonicity:} if $X\prec _{SSD}Y$, then $\rho(X)\leq \rho(Y)$.\\

(C3) \textit{CX-monotonicity:} if $X\prec_{CX}Y$, then $\nu(X)\leq\nu(Y)$.
\end{dfn}

\begin{rmk}
Let $q\in [1,\infty]$, on $L^{q}$ all real-valued law-invariant coherent risk measures are SSD-monotone. We refer the reader to Dana \cite{5}, Grechuk et al. \cite{14}, and F\"{o}llmer and Schied \cite{9} for proofs of the above assertions, and to Mao and Wang \cite{15} for a characterization of SSD-monotone risk measures.
\end{rmk}

\section{Extended Gini Shortfall}
\label{sec3.3}

In this section, we expose our main contribution, which is based on the Gini coefficient, a free-center measure of variability that was introduced by Corrado Gini as an alternative to the variance measure (e.g., Giorgi (\cite{12}, \cite{13}) and Ceriani and Verme \cite{3}). The Gini coefficient has been remarkably influential in numerous research areas (e.g, Yitzhaki and Schechtman \cite{30} and the references therein). Yitzhaki \cite{28} lists more than a dozen alternative presentations of the Gini coefficient. We now present a formal definition.

\begin{dfn}
The Gini coefficient is a functional $Gini:L^1\rightarrow[0,\infty]$ defined conform:
\begin{equation}
  Gini(X)=\mathbb{E}[\mid X^{*}-X^{**}\mid],
\end{equation}
where $X^{*}$ and $X^{**}$ are two independent copies of $X$.
\end{dfn}

\begin{rmk}
The Gini coefficient can be written in terms of a signed Choquet integral:
\begin{equation}\label{4}
  Gini(X)=2\int_{0}^{1}F^{-1}_{X}(u)(2u-1)du.
\end{equation}
From Theorem \ref{3}, it follows immediately that the Gini coefficient is a coherent measure of variability and it is CX-monotone. Moreover, equation (\ref{4}) can be written in terms of covariance (which is the most common formula of the Gini coefficient):
\begin{equation}
  Gini(X)=4Cov[F^{-1}_{X}(U),U]=4 Cov[X,U_{X}].
\end{equation}
We recall that $U$ can be any uniformly on $[0,1]$ distributed rv, and
$U_{X}$ is a uniform $[0,1]$ rv such that the equation $F^{-1}_{X}(U_{X})=X$ holds.
\end{rmk}

The Gini functional supposes that all individuals have the same attitude towards risk. Nonetheless, the concept can be extended into a family of measures of variability differing from each other in the decision-maker's degree of risk aversion, which is reflected in this paper by the parameter $r$. The basic definition of the Extended Gini coefficient is based on the covariance term. We refer to Yitzhaki \cite{27}, Shalit and Yitzhaki \cite{24}, Yitzhaki and Schechtman \cite{29}, and Yitzhaki and Schechtman \cite{30} for an overview of the Extended Gini properties. We now expose it in a formal sense.

\begin{dfn}
The Extended Gini coefficient is a functional $EGini_{r}:L^1\rightarrow[0,\infty]$ defined conform:
\begin{equation}
  EGini_{r}(X)= -2rCov[X,(1-F_{X}(X))^{r-1}],\:r>1.
\end{equation}
\end{dfn}

\begin{rmk}
There are special cases of interest for the Extended Gini:
\begin{description}
  \item[$\centerdot$] For $r=2$: the Extended Gini coefficient becomes the simple Gini.
  \item[$\centerdot$] For $r\rightarrow \infty$: the Extended Gini reflects the attitude of a max-min investor who expresses risk only in terms of the worst outcome.
  \item[$\centerdot$] For $r\rightarrow 1$: the Extended Gini tends to zero and represents the attitude of a risk-neutral individual who does not care about variability.
\end{description}
\end{rmk}

We now explore the characterization of the Extended Gini coefficient as a signed Choquet integral. In this sense, note that from equation (\ref{5}) in Theorem \ref{3}, if one sets $g_{r}(u)=-r (1-u)^{^{r-1}}$ for $r>1$ and $u\in [0,1]$,
we run into $  \nu (X)=-r \,Cov[X,(1-F_{X}(X))^{r-1}]$. With that in mind, we now state and prove the formal result.

\begin{prp}\label{prp1} The Extended Gini functional is a CX-monotone coherent measure of variability, represented by the signed Choquet integral
\begin{equation}
  EGini_{r}(X)= 2\int_{0}^{1}F^{-1}_{X}(u)(1+g_{r}(u))du.
\end{equation}
\end{prp}
\begin{proof}
 We recall that $U$ can be any uniformly distributed rv on $[0,1]$ such that the equation $F^{-1}_{X}(U)=X$ holds, where $\mathbb{E}[X]=m$. In order to obtain the proof, we need the claim that $\mathbb{E}[(1-F_{X,p}(X))^{r-1}|X>x_{p}]= (1-p)^{r-1}/r$ for $X\in L^{1}$, $r\in (1,\infty)$ and $p\in [0,1)$. In order to prove it, we get that
  \begin{alignat*}{3}
    \mathbb{E}[(1-F_{X,p}(X))^{r-1}|X>x_{p}] &= \mathbb{E}[(1-U)^{r-1}|U>p] \\
     &= \frac{1}{1-p}\int_{\mathbb{R}}(1-u)^{r-1}\mathds{1}_{[p,1]}(u) du\\
     &= \frac{1}{1-p}\int_{p}^{1}(1-u)^{r-1} du\\
     &= \frac{1}{1-p}[-(1-u)^{r}/r]_{p}^{1}\\
     &= (1-p)^{r-1}/r.
  \end{alignat*}
 Under this perspective, we can easily verify that: $\mathbb{E}[(1-F_{X}(X))^{r-1}]=1/r$. Thus, we get the following:
  \begin{alignat*}{2}
   EGini_{r}(X)&= -2r\, Cov[X,(1-F_{X}(X))^{r-1}] \\
    &= -2r\, \mathbb{E}[(X-\mathbb{E}(X))((1-F_{X}(X))^{r-1}-\mathbb{E}[(1-F_{X}(X))^{r-1}])]\\
    &= -2r\, \mathbb{E}[(F^{-1}_{X}(U)-m)((1-U)^{r-1}-\frac{1}{r})]\\
    &= -2r\, \int_{0}^{1}(F^{-1}_{X}(u)-m)((1-u)^{r-1}-\frac{1}{r})du\\
    &= -2r\,\int_{0}^{1}F^{-1}_{X}(u)((1-u)^{r-1}-\frac{1}{r})du  \\
    &= 2\,\int_{0}^{1}F^{-1}_{X}(u)(1+g_{r}(u))du.
  \end{alignat*}
 From this signed Choquet integral representation, the other claims in the proposition follow immediately from Theorem \ref{3} and the fact that all coherent measures of variability are CX-monotone. This concludes the proof.
\end{proof}

\begin{rmk}
In this case, $h_{r}$ in equation (\ref{6}) of Theorem \ref{3} is given by:
\begin{equation}
  h_{r}(u)=u+(1-u)^{r}-1,\:r>1
\end{equation}
\end{rmk}

Based on the exposed content, we now turn the focus to an adaptation to the tails of the distribution function. In this sense, we now introduce the Tail Extended Gini (TEG) functional, as well formally prove its properties and Choquet integral representation.

\begin{dfn}
The Tail Extended Gini is a functional $TEGini_{r,p}:L^1\rightarrow[0,\infty]$ defined conform:
\begin{equation}
  TEGini_{r,p}(X)= \frac{-2r}{1-p}\, Cov[X,(1-F_{X}(X))^{r-1}| X>x_{p}],\:r>1,\:0<p<1.
\end{equation}
\end{dfn}

\begin{prp}
 The Tail Extended Gini is standardized, location invariant, positively homogeneous and co-monotonic additive. Moreover, it is a signed Choquet integral conform:
 \begin{equation}
  TEGini_{r,p}(X)=\frac{2}{(1-p)^{2}}\int_{p}^{1} F_{X}^{-1}(u) [g_{r}(u)+(1-p)^{r-1}]du.
\end{equation}
\end{prp}

\begin{proof}
The fact that the Tail Extended Gini is standardized, location invariant, positively homogeneous is easily realized from its definition. For the Choquet representation, we again recall that $U$ can be any uniformly distributed rv on $[0,1]$ such that the equation $F^{-1}_{X}(U)=X$ holds, where $\mathbb{E}[X]=m$. Thus, we obtain that
\begin{alignat*}{3}
  TEGini_{r,p}(X)&= \frac{-2r}{1-p}\, Cov[X,(1-F_{X}(X))^{r-1}|X>x_{p}]\\
                 &= \frac{-2r}{1-p}\, \mathbb{E}[(X-m)((1-F_{X}(X))^{r-1}-(1-p)^{r-1}/r)| X>x_{p}]\\
                 &= \frac{-2r}{1-p}\, \mathbb{E}[(X-m)((1-F_{X}(X))^{r-1}-(1-p)^{r-1}/r)| X>x_{p}]\\
                 &= \frac{-2r}{1-p}\, \mathbb{E}[(F_{X}^{-1}(U)-m)((1-U)^{r-1}-(1-p)^{r-1}/r)| U>p]\\
                 &= \frac{-2r}{(1-p)^{2}}\, \int_{p}^{1}(F^{-1}_{X}(u)-m)((1-u)^{r-1}-(1-p)^{r-1}/r)du\\
                 &= \frac{2}{(1-p)^{2}}[\int_{p}^{1} F^{-1}_{X}(u)(-r(1-u)^{r-1}+(1-p)^{r-1})du\\
                 &\qquad\qquad + m r\,\int_{p}^{1}((1-u)^{r-1}-(1-p)^{r-1}/r)du]\\
                 &=\frac{2}{(1-p)^{2}}\,\int_{p}^{1} F^{-1}_{X}(u)[g_{r}(u)+(1-p)^{r-1}]du.
\end{alignat*}
In fact, from the proof of the previous proposition we have that
$\int_{p}^{1}((1-u)^{r-1}-(1-p)^{r-1}/r)du=0$.
 Finally, the Choquet representation implies co-monotonic additivity. This completes the proof.
\end{proof}

\begin{rmk}
 However, as shown in a counter example (for $r=2$) by Furman et al. \cite{11}, the Tail Extended Gini is not sub-additive. Therefore, unlike the Extended Gini functional, the tail counterpart is not a coherent measure of variability.
\end{rmk}

Despite the fact that TEG is not a coherent measure of variability, we will show now that a linear combination of the Expected Shortfall with the Tail Extended Gini gives rise to a coherent risk measure, the Extended Gini Shortfall, that quantifies both the magnitude and the variability of tail risks. We now define such a combination.

\begin{dfn}
Let $p\in (0,1)$. We have that the Value at Risk and the Expected Shortfall are functionals $VaR_p:L^0\rightarrow(-\infty,\infty]$ and $ES_p:L^1\rightarrow(-\infty,\infty ]$ defined conform:
\begin{equation}
  VaR_{p}(X)=\inf\{x\in \mathbb{R}: F_{X}(x)\geq p\},
\end{equation}
\begin{equation}
  ES_{p}(X)= \frac{1}{1-p}\int_{p}^{1}VaR_{q}(X)dq.
\end{equation}
\end{dfn}

\begin{rmk}
It is a well-established fact that ES is a SSD-monotone comonotonic additive coherent risk measure. When the cdf $F_{X}$ is continuous, ES coincides with tail conditional expectation $\mathbb{E}[X| X\geq x_{p}]$.
\end{rmk}

\begin{dfn}
The Extended Gini Shortfall is a functional $EGS_{r,p}^{\lambda}:L^1\rightarrow(-\infty,\infty]$ defined conform:
\begin{equation}
 EGS_{r,p}^{\lambda}(X)= ES_{p}(X)+\lambda\, TEGini_{r,p}(X), \lambda\geq0.
\end{equation}
\end{dfn}

In order to be a reasonable tool for risk management, the properties of coherent risk measures are desired. However, as mentioned in the previous section, TEG is not sub-additive, and as a measure of variability is not monotone.  However, when $\lambda$ is zero, then EGS obviously inherits all the properties of the ES which is coherent, but when $\lambda$ is sufficiently large, then the TEG starts to dominate ES, and thus coherence of EGS cannot be expected. Intuitively, as suggested by Furman et al. \cite{11}, there might be a threshold that delineates the value of $\lambda$ for which EGS is coherent. We now verify it in a formal way.

\begin{prp}
  The Extended Gini shortfall is translation invariant, positively homogeneous, and co-monotonically additive, which can be represented as a signed Choquet integral conform:
  \begin{equation}
    EGS_{r,p}^{\lambda}(X)=\int_{0}^{1}F^{-1}_{X}(u)\,\phi_{r,p}^{\lambda}(u) du
  \end{equation}
  where,
  \begin{equation}
   \phi_{r,p}^{\lambda}(u)=\frac{1}{(1-p)^{2}}[1-p+2\lambda(g_{r}(u)+(1-p)^{r-1})]\,\mathds{1}_{[p,1]}(u), \quad u\in [0,1].
  \end{equation}
 Moreover, the Extended Gini Shortfall is a SSD-monotone coherent risk measure for $\lambda \in [0\,,\, 1/(2(r-1)(1-p)^{r-2})]$.
\end{prp}
\begin{proof}
The translation invariance, positive homogeneity and co-monotonic additivity are easily verifiable from the properties of ES and TEG. Regarding the Choquet representation, we have that:

  \begin{alignat*}{4}
    EGS_{r,p}^{\lambda}(X)&= ES_{p}(X)+\lambda\, TEGini_{r,p}(X)\\
     & = \frac{1}{1-p}\int _{p}^{1} F^{-1}_{X}(u)du+ \frac{2\lambda}{(1-p)^{2}}\,\int_{p}^{1} F^{-1}_{X}(u)[g_{r}(u)+(1-p)^{r-1}]du\\
     & = \int_{p}^{1} F^{-1}_{X}(u)[\frac{1}{1-p}+\frac{2\lambda}{(1-p)^{2}}(g_{r}(u)+(1-p)^{r-1})]du\\
     & = \int_{0}^{1}F^{-1}_{X}(u)\,\phi_{r,p}^{\lambda}(u) du.
  \end{alignat*}
 For coherence, it remains to prove that EGS is monotone and sub-additive for $\lambda \in [0\,,\, 1/(2(r-1)(1-p)^{r-2})]$. Note that $\phi_{r,p}^{\lambda}$ is an increasing function on $[0,1]$, therefore $\phi_{r,p}^{\lambda}(u)$ is non-negative if and only if $\phi_{r,p}^{\lambda}(p)\geq 0$. Thus, $\phi_{r,p}^{\lambda}$ is non-negative if and only if \\$\lambda\in [0\,,\, 1/(2(r-1)(1-p)^{r-2})]$. This fact implies monotonicity and sub-additivity from the properties discussed on section 2. Thus, EGS is a coherent risk measure for this choice of $\lambda$. Finally, SSD-monotonicity for this choice of $\lambda$, it is directly implied by the fact that EGS is law invariant. This concludes the proof.
\end{proof}

\begin{rmk}
From the previous Proposition we can link the EGS with its acceptance set, which is defined as $\mathcal{A}_{EGS_{r,p}^{\lambda}}=\{X\in L^1 : EGS_{r,p}^{\lambda}(X)\leq 0\}$. It is well-known, see F\"{o}llmer and Schied \cite{9} for instance, that this set is convex, law invariant, monotone, closed for multiplication with positive scalar and addition between co-monotonic variables. Moreover, we have that $\mathcal{A}_{EGS_{r,p}^{\lambda}}$ contains $L^1_+$ and has no intersection with $\{X\in L^1 :X\notin L^1_+\}$. It is direct to verify, from Translation Invariance of EGS, that $EGS_{r,p}^{\lambda}(X)=\inf\{m : X+m\in \mathcal{A}_{EGS_{r,p}^{\lambda}}\}$.
\end{rmk}

\begin{rmk}
We have that EGS can be represented as convex combination of ES at distinct levels of $p$, a Kusuoka representation, conform $EGS_{r,p}^{\lambda}(X)=\int_{0}^{1}ES_p(X)\mu(dp)$, where $\phi_{r,p}^{\lambda}(u)=\int_{[1-u,1)}\frac{1}{s}\mu(ds)$. Here, $\mu$ is a probability measure over $(0,1]$. These representations are linked to the well-known dual representation, conform $EGS_{r,p}^{\lambda}(X)=E[XQ]$, where $F^{-1}_{Q}(u)=\phi_{r,p}^{\lambda}(u)$. We can think about $Q$ as the relative density (Radon-Nikodym) of an alternative probability measure absolutely continuous in relation to $\mathbb{P}$.
\end{rmk}

\begin{rmk}
From the previous Proposition, we get that the Extended Gini Shortfall is part of the spectral risk measures class, introduced in Acerbi \cite{1}, characterized by the weighting function $\phi_{r,p}^{\lambda}$, which enables to reflect the individual's subjective attitude toward risk. In Furman et al. \cite{11} a specific case is introduced ($r=2$), which assigns the same weighting function to all decision-makers. Thus, there is a connection to the individual's risk aversion function. As a result, it is more legitimate for $\phi_{r,p}^{\lambda}$ to be dependent on the parameters $r$, $p$ and $\lambda$.
\end{rmk}

We now provide a result and interpretation about how this weighting function $\phi$ behaves in relation to changes (partial derivatives) of each variable (parameter) among $u$, $r$, $p$ and $\lambda$. It is valid to point out that, by the spectral representation, results can be directly understood as the effect each parameter has over values for EGS.

\begin{prp}
Consider the weighting function
\begin{equation*}
\phi_{r,p}^{\lambda}(u)=\phi(u,r,p,\lambda)=\frac{1-p+2\lambda[(1-p)^{r-1}-r(1-u)^{r-1}]}{(1-p)^{2}}\,\mathds{1}_{[p,1]}(u),
\end{equation*}
where $u\in [0,1]$, $p\in(0,1)$, $r>1$, and $\lambda\in [0, 1/(2(r-1)(1-p)^{r-2})]$. We have that:\\

(i) The interval for values of $\lambda$ that make EGS subadditive has superior limit non-decreasing in $p$, if $r\geq2$ and non-increasing in $p$ otherwise. Moreover, the interval is non-decreasing in $r$ if, and only if, $r\geq 1-\frac{1}{\ln(1-p)}$;\\

(ii) $\phi$ is non-decreasing in $u$;\\

(iii) $\phi$ is non-decreasing in $\lambda$;\\

(iv) $\phi$ is non-decreasing in $p$ if, and only if
\begin{equation*}
u\geq1-\left(\frac{1-p-2\lambda(r-3)(1-p)^{r-1}}{4\lambda r} \right) ^{\left(\frac{1}{r-1} \right) };
\end{equation*}\\
	
(v) $\phi$ is non-decreasing in $r$ if, and only if
\begin{equation*}
(1-p)^{(1-p)}\geq\left( \exp\left\lbrace (1-u)^{r-1}[r\ln(1-u)+1] \right\rbrace\right)  ^{\left(\frac{1}{r-1} \right) }.
\end{equation*}
\end{prp}

\begin{proof}
For (i), we must remember that EGS is subadditive when $\lambda\in [0, 1/(2(r-1)(1-p)^{r-2})]$. Consider the functional
\begin{equation*}
B(r,p)= \frac{1}{2(r-1)(1-p)^{r-2}}
\end{equation*}
that represents the threshold that $\lambda$ shall not exceed. Regarding $p$, we get
\begin{equation*}
\frac{\partial B}{\partial p}(r,p)=B'_{p}(r,p)= \frac{(r-2)(1-p)^{1-r}}{2(r-1)}.
\end{equation*}
It is direct that the sign of $B'_{p}$ depends on the sign of $r-2$. Thus, $B(r,p)$ is non-decreasing in $p$ for $r\geq2$ and non-increasing otherwise.
Regarding $r$, we thus have that
\begin{equation*}
\frac{\partial B}{\partial r}(r,p)=B'_{r}(r,p)= \frac{-1}{2}\frac{(1-p)^{r-2}[1+(r-1)\ln(1-p)]}{[(r-1)(1-p)^{r-2}]^{2}}.
\end{equation*}
The sign of $B'_{r}$ depends on the sign of $
C_{p}(r)= 1+(r-1) \ln(1-p)$. The decreasing function $C_{p}$ maps $(1,\infty)$ to $(-\infty,1)$, then there exists a unique critical value $r_{0}=1-\frac{1}{\ln(1-p)}$ such that $B_{p}$ non-increases over $(1,r_{0}]$ and non-decreases on $(r_{0},\infty)$.

For (ii), the claim follows from the fact that EGS is a spectral risk measure. More specifically, we have that
\begin{equation*}
\frac{\partial\phi}{\partial u}(u,r,p,\lambda)=\frac{2\lambda\, r(r-1)(1-u)^{r-2}}{(1-p)^2}\,\mathds{1}_{[p,1]}(u),
\end{equation*}
which is non-negative any case.

Regarding (iii), the claim is direct by definition of EGS. More specifically on the weighting function we obtain
\begin{equation*}
\frac{\partial\phi}{\partial\lambda}(u,r,p,\lambda)=\frac{2[(1-p)^{r-1}-r(1-u)^{r-1}]}{(1-p)^2}\,\mathds{1}_{[p,1]}(u).
\end{equation*}
We can note that this expression is non-decreasing in $u$ with critical value when $u=1-(1-p)r^{-\frac{1}{r-1}}$. Thus, $\frac{\partial\phi}{\partial\lambda}\geq0$ when $u\geq1-(1-p)r^{-\frac{1}{r-1}}\geq p$, since $r^{\frac{1}{r-1}}\geq 1$. But $u\geq p$ is the only case that matters because of the indicator function in  $\frac{\partial\phi}{\partial\lambda}$.

For item (iv) we begin by noticing that
\begin{equation*}
\frac{\partial\phi}{\partial p}(u,r,p,\lambda)=\frac{1-2\lambda(r-3)(1-p)^{r-2}-4\lambda r(1-u)^{r-1}(1-p)^{-1}}{(1-p)^2}\,\mathds{1}_{[p,1]}(u),
\end{equation*}
which is a non-decreasing expression in $u$. Thus, we have that $\frac{\partial\phi}{\partial p}\geq0$ if, and only if,
\begin{equation*}
u\geq1-\left(\frac{1-p-2\lambda(r-3)(1-p)^{r-1}}{4\lambda r} \right) ^{\frac{1}{r-1} }.
\end{equation*}

Finally, regarding item (v), we follow the same reasoning to get
\begin{equation*}
\frac{\partial\phi}{\partial r}(u,r,p,\lambda)=\frac{2\lambda[\ln(1-p)(1-p)^{r-1}-(1-u)^{r-1}-r\ln(1-u)(1-u)^{r-1}]}{(1-p)^2}.
\end{equation*}
This is a complex expression which can assume both positive and negative values since some terms in numerator posses distinct signs without a domination. Moreover, isolating $u$ or $r$ is not trivial. Nonetheless, after some manipulation we obtain that $\frac{\partial\phi}{\partial r}\geq0$ if, and only if, $(1-p)^{(1-p)}\geq\left( \exp\left\lbrace (1-u)^{r-1}[r\ln(1-u)+1] \right\rbrace\right)  ^{\left(\frac{1}{r-1} \right) }$. This concludes the proof.
\end{proof}

\begin{rmk}
The non-linear behavior of the partial derivatives found by the previous proposition is related to the function $g_r(u)=-r(1-u)^{r-1}$, which is central to the theory we develop, which is not trivial for $r\neq2$, the case considered in Furman et al. \cite{11}. As one can easily note, it recurrently appears in the expression for partial derivatives. Hence, our contribution is by extending the canonical case to a situation where more complexity regarding $r$ can be addressed.
\end{rmk}

 Since $p$ reflects the prudence level, which is usually close to $1$ in practice, then $r_{0}$ in the proof of item (i) is in general close to 1. Thus, for practical matters, the superior limit for $\lambda$ is non-decreasing in most relevant values of $r$. The most complex parameter sensitivities are regarding the prudence level $p$ and the generalization term $r$, respectively in items (iv) and (v), because both $\frac{\partial\phi}{\partial p}$ and $\frac{\partial\phi}{\partial r}$ are expressions that can assume both positive and negative values since some terms in numerator posses distinct signs without a domination. This can be explained due to the fact that $\phi$ is a weighting function and changes in $p$ and $r$ alter how much `mass' is put to any probability level $u$. Because $\int_{0}^{1}\phi_{r,p}^{\lambda}(u)du=1$, it is necessary that the increase on $\phi(u)$ for some values of $u$ be compensated by some decrease in others.

 Regarding prudence level, item (iv) in the previous proposition emphasizes that $\phi$ non-decreases in $p$ for larger values of $u$. This is in consonance with practical intuition because more weight is put to extreme probabilities. Moreover, this is corroborated when we verify the variations of $\phi$ regarding $u$ and $p$, where we obtain the non-negative expression
 \begin{equation*}
 \frac{\partial^2\phi}{\partial u\partial p}=\frac{4\lambda r(r-1)(1-u)^{r-2}}{(1-p)^3}\,\mathds{1}_{[p,1]}(u).
 \end{equation*}

  When we consider the sensibility of variations regarding $u$ and $r$, in a case we get
 \begin{equation*}
\frac{\partial^2\phi}{\partial u\partial r}=\frac{2\lambda(1-u)^{r-2}[(2r-1)+(r^2-r)\ln(1-u)]}{(1-p)^2}\,\mathds{1}_{[p,1]}(u),
 \end{equation*}
 there is divergence about sign -- first term inside brackets is non-negative while the second is non-positive, corroborating with the previous argument. Nonetheless, $\phi$ is non-decreasing in $r$, a situation where $r$ behaves more like a risk aversion coefficient, when a non-decreasing function of $p$, $(1-p)^{(1-p)}$ is greater than a threshold that is an expression depending on both $u$ and $r$. Repeating the argument of practical choices for the prudence level, $p$ will be close to $1$ and $\phi$ will be increasing in $r$ for more values of $u$. In this sense, $r$ must be understood as a generalization parameter for a family of EGS risk measures rather than a single linear risk aversion coefficient.

\section{Extended Gini Shortfall for usual distributions }

In this section we provide analytical formulations to compute the proposed Extended Gini Shortfall to known and very used distribution functions. In that sense, a location-scale family is a family of probability distributions parameterized by a location parameter and a non-negative scale parameter. Suppose that $Z$ is a fixed rv taking values in $\mathbb{R}$. For $\alpha \in \mathbb{R}$ and $\beta \in (0,\infty)$, let $X=\alpha+\beta Z$. The two-parameter family of distributions associated with $X$ is called the \textit{location-scale family} associated with the given distribution of $Z$; $\alpha$ is called the \textit{location parameter} and $\beta$ the \textit{scale parameter}. The standard form of any distribution is the form whose location and scale parameters are $0$ and $1$, respectively. In this section, we restrain our attention into standardized rv's. In general, when $X=\alpha+\beta Z$ for $\alpha\in \mathbb{R}$ and $\beta\in (0,\infty)$, we have, directly from their properties, both $ES_{p}(X)=\alpha+\beta ES_{p}(Z)$ and $TEGini_{r,p}(X)=\beta TEGini_{r,p}(Z)$.

In what follows, we start with the general elliptical family and then specialize the obtained result to uniform, normal and Student-t distribution cases. We recall that $X$ is an elliptical distribution if $X \stackrel{d}{=}\alpha+\beta Z$, where $Z$ is a spherical distribution. Let $Z$ be a spherical rv with characteristic generator $\psi : [0,\infty)\rightarrow \mathbb{R}$; succinctly $Z\sim S(\psi)$. When $Z$ has a probability density function (pdf), then there is a density generator $g :[0,\infty)\rightarrow [0,\infty)$ such that $\int _{0}^{\infty}z^{-1/2}g(z) dz<\infty$, we succinctly write $Z\sim S(g)$. We can express the pdf $f : \mathbb{R}\rightarrow [0,\infty]$ of $Z$ by $f(z)= c\, g(z^{2}/2)$,
where $c>0$ is the normalizing constant. The mean $\mathbb{E}[Z]$ is finite when $\int _{0}^{\infty} g(z)dz<\infty$, in which case we have $\mathbb{E}[Z]=0$ because the pdf $f$ is symmetric around $0$. Under this condition, we define the function $\overline{G}: [0,\infty)\rightarrow [0,\infty)$ by
 $\overline{G}(y)=c \int_{y}^{\infty}g(x)dx$, which is called the tail generator of $Z$. We now state and prove the main result in this section.

\begin{prp}
  Let $Z\sim S(g)$, $\mathbb{E}[Z]$ finite, and $p\in (0,1)$. Then, we have:
  \begin{equation}\label{7}
    ES_{p}(Z)=\frac{\overline{G}(z_{p}^{2}/2)}{1-p},
  \end{equation}
  \begin{equation}\label{8}
    TEGini_{r,p}(Z)=\frac{2r(r-1)}{1-p}\mathbb{E}\left[(1-F_{Z}(Z))^{r-2}\overline{G}(Z^{2}/2)|Z>z_{p}\right] + 2[1-r(1-p)^{r-2}]ES_{p}(Z).
  \end{equation}
\end{prp}
\begin{proof}
For ES. we have that:
  \begin{alignat*}{4}
    ES_{p}(Z) &= \mathbb{E}[Z|Z>z_{p}] \\
     &= \frac{1}{1-p} \int_{z_{p}}^{\infty}z f(z)dz\\
     &= \frac{c}{1-p} \int_{z_{p}}^{\infty}z g(z^{2}/2)dz\\
     &= \frac{c}{1-p} \int_{z_{p}^{2}/2}^{\infty} g(x)dx\\
     &= \frac{\overline{G}(z_{p}^{2}/2)}{1-p}.
  \end{alignat*}
 Concerning to TEG, we get:
  \[TEGini_{r,p}(Z)=\frac{-2r}{1-p}\left[\mathbb{E}[Z(1-F_{Z}(Z))^{r-1}|Z>z_{p}]-\mathbb{E}[Z|Z>z_{p}]\mathbb{E}[(1-F_{Z}(Z))^{r-1}|Z>z_{p}]\right]\]
  \begin{alignat*}{2}
    \mathbb{E}[Z(1-F_{Z}(Z))^{r-1}|Z>z_{p}] & = \frac{1}{1-p}\mathbb{E}[Z(1-F_{Z}(Z))^{r-1}\,\mathds{1}_{\{Z>z_{p}\}}] \\
     &= \frac{1}{1-p}\int_{z_{p}}^{\infty}z(1-F_{Z}(z))^{r-1}f(z)dz
  \end{alignat*}
  Note that $zf(z)dz=-d\overline{G}(z^{2}/2)$ and  $\overline{G}(z^{2}/2)=(1-p)ES_{p}(Z)$. Integration by parts leads to:
  \begin{alignat*}{2}
    \mathbb{E}[Z(1-F_{Z}(Z))^{r-1}|Z>z_{p}] & = \frac{1}{1-p}(1-p)^{r-1}\overline{G}(z^{2}/2)-(r-1)\mathbb{E}[(1-F_{Z}(Z))^{r-2}\overline{G}(Z^{2}/2)|Z>z_{p}] \\
     & = (1-p)^{r-1}ES_{p}(Z)-(r-1)\mathbb{E}[(1-F_{Z}(Z))^{r-2}\overline{G}(Z^{2}/2)|Z>z_{p}].
  \end{alignat*}
 \[ \mathbb{E}[Z|Z>z_{p}]=ES_{p}(Z).\]
 Finally, similarly to the proof of Proposition \ref{prp1}, we have:
 \[\mathbb{E}[(1-F_{Z}(Z))^{r-1}|Z>z_{p}]=\frac{(1-p)^{r-1}}{r}.\]
 This completes the proof.
\end{proof}
\begin{rmk}\label{rmk}
Note that $Var(Z)$ is finite whenever $\int_{0}^{\infty}z^{1/2}g(z) dz<\infty$, in which case $Var(Z)$ is equal to $\int_{-\infty}^{\infty}\overline{G}(z^{2}/2)dz$ (by an integration by parts). Hence, we have that $f^{*}(z)=\overline{G}(z_{p}^{2}/2)/Var(Z)$ is a pdf. In this case we get the expression $ TEGini_{r,p}(Z)=\frac{2r(r-1)}{1-p}Var(Z) \mathbb{E}\left[(1-F_{Z}(Z))^{r-2}f^{*}(Z)|Z>z_{p}\right] + 2[1-r(1-p)^{r-2}]ES_{p}(Z)$.
\end{rmk}

We now focus on the case of a standard uniform rv $Z\sim U[-1,1]$. In this sense, we have the following Corollary.

\begin{cor}
  Let $Z\sim U[-1,1]$ and $p\in(0,1)$. Then, we have:
  \begin{equation}
    ES_{p}(Z)=\frac{1-z_{p}^{2}}{4(1-p)},
  \end{equation}
\begin{equation}
  TEGini_{r,p}(Z)=\frac{r}{3(1-p)^{2}}\left(\frac{1-z_{p}}{2}\right)^{r-1}+ 2[1-r(1-p)^{r-2}]ES_{p}(Z).
\end{equation}
\end{cor}
\begin{proof}
  The standard uniform is a spherical distribution with density $f_{Z}(z)=\frac{1}{2}\mathds{1}_{[-1,1]}(z)= \frac{1}{2}\mathds{1}_{[0,1/2]}(z^{2}/2),\quad z\in[-1,1]$. Thus, we obtain
$g(z)=\mathds{1}_{[0,1/2]}(z)\quad \text{ and }\quad c=\frac{1}{2}$. In this case, the tail generator is given by:
\[\overline{G}(y)=\int_{y}^{\infty}\frac{1}{2}g(t)dt=\frac{1}{4}-\frac{y}{2}.\]
Hence, we obtain \[\overline{G}(z_{p}^{2}/2)=\frac{1-z_{p}^{2}}{4},\:ES_{p}(Z)=\frac{1-z_{p}^{2}}{4(1-p)}\]

In order to prove the result for TEGini, we use Remark \ref{rmk}. Thus, we obtain:
\[TEGini_{r,p}(Z)=\frac{2r(r-1)}{3(1-p)} \mathbb{E}\left[\left(\frac{1-Z}{2}\right)^{r-2}f^{*}(Z)|Z>z_{p}\right] + 2[1-r(1-p)^{r-2}]ES_{p}(Z),\]
with \begin{alignat*}{3}
  \mathbb{E}\left[\left(\frac{1-Z}{2}\right)^{r-2}f^{*}(Z)|Z>z_{p}\right] &= \frac{1}{4(1-p)}\int_{z_{p}}^{1}\left(\frac{1-z}{2}\right)^{r-2}dz \\
   &= \frac{1}{2(1-p)(r-1)}\left( \frac{1-z_{p}}{2}\right)^{r-1}.
\end{alignat*}
This completes the proof.
\end{proof}

In this next corollary, we deal with the standard Normal rv $Z\sim \mathcal{N}(0,1)$ whose cdf we denote by $\Phi$.

\begin{cor}
  Let $Z\sim \mathcal{N}(0,1)$ and $p\in (0,1)$. Then, we have:
  \begin{equation}\label{9}
    ES_{p}(Z)=\frac{\Phi'(z_{p})}{1-p},
  \end{equation}
   \begin{equation}\label{10}
  TEGini_{r,p}(Z)=\frac{2r(r-1)}{1-p} \mathbb{E}\left[(1-\Phi(Z))^{r-2}\Phi'(Z)|Z>z_{p}\right] + 2[1-r(1-p)^{r-2}]ES_{p}(Z)
\end{equation}
\end{cor}

\begin{proof}
  The standard normal is a spherical distribution with density generator $ g(z)=exp(-z),\, c=1/\sqrt{2\pi}\, \text{ and }\, \overline{G}(z^{2}/2)=\Phi'(z)$. Thus, equation (\ref{9}) follows immediately from equation (\ref{7}). To establish the second part, we use Remark \ref{rmk} with $Var(Z)=1$.
\end{proof}

Finally, we expose a corollary concerning to the Student-t distribution. Let $n\in \mathbb{R}_{+}^{*}$, we say that $Z$ has a standard Student-t distribution with $n$ degree of freedom if its pdf is:
\[f_{n}(z)=\frac{1}{\sqrt{n}\,Beta(n/2,1/2)}\left(1+\frac{z^{2}}{n}\right)^{-(\frac{n+1}{2})},\:z\in \mathbb{R}\]
Where $Beta(a,b)$ stands for the Beta distribution. We set $\theta=\frac{n+1}{2}$ and $k_{\theta}=\frac{n}{2}$. Then, we denote $Z\sim t(\theta)$ with parameter $\theta>\frac{1}{2}$ and the pdf can be rewritten as:
\[f_{\theta}(z)=c_{\theta}\left(1+\frac{z^{2}}{2k_{\theta}}\right)^{-\theta},\:z\in \mathbb{R}\]
where $c_{\theta}=(\sqrt{2k_{\theta}}Beta(\theta-1/2,1/2))^{-1}$.
The expected value of $Z$ is well-defined only for $\theta >1$ and the variance of $Z$ is finite only if $\theta > \frac{3}{2}$. Let $F_{\theta}$ be the cdf of $Z$.

\begin{cor}
  Let $Z\sim t(\theta)$, $\theta> 1$, and $p\in (0,1)$. Then, we have:
  \begin{equation}
    ES_{p}(Z)= \frac{c_{\theta}k_{\theta}}{(1-p)(\theta-1)}\left(1+\frac{z^{2}_{p}}{2k_{\theta}}\right)^{-(\theta-1)},
  \end{equation}
  \begin{align}
    TEGini_{r,p}(Z) &=\frac{2r(r-1)}{(1-p)(\theta-1)}\frac{c_{\theta}k_{\theta}}{c_{\theta-1}} \mathbb{E}\left[(1-F_{\theta}(Z))^{r-2}f_{\theta-1}\left(\sqrt{\frac{k_{\theta-1}}{k_{\theta}}}Z\right)|Z>z_{p}\right]\newline\nonumber\\
    &+ 2[1-r(1-p)^{r-2}]ES_{p}(Z).
  \end{align}
\end{cor}

 \begin{proof}
   The tail generator of the standard Student-t is given by:
   \[g(z)=\left(1+\frac{z}{k_{\theta}}\right)^{-\theta},\qquad z\in \mathbb{R}\]
   Hence, we get that: \begin{align*}
            \overline{G}(z) &=c_{\theta}\int_{z}^{\infty}\left(1+\frac{t}{k_{\theta}}\right)^{-\theta}dt \\
             &= \frac{c_{\theta}k_{\theta}}{\theta-1}\left(1+\frac{z}{k_{\theta}}\right)^{-(\theta-1)},
          \end{align*}
   which leads to: \[\overline{G}(Z^{2}/2)=\frac{c_{\theta}k_{\theta}}{c_{\theta-1}(\theta-1)}f_{\theta-1}\left(\sqrt{\frac{k_{\theta-1}}{k_{\theta}}}Z\right)\]
   Then, equation (\ref{10}) follows immediately from equation (\ref{8}).
 \end{proof}

\section{From theory to practice}
In this section we are going to provide an illustration for the practical usage of $EGS$. In this sense, consider the $EGS_{r,p}^{\lambda}$ definition:
\begin{equation*}
  EGS_{r,p}^{\lambda}(X)= \int_{0}^{1}F^{-1}_{X}(u)\,\phi_{r,p}^{\lambda}(u) du,
\end{equation*}
where,\begin{equation*}
\phi_{r,p}^{\lambda}(u)= \frac{1-p+2\lambda[(1-p)^{r-1}-r(1-u)^{r-1}]}{(1-p)^{2}}\,\mathds{1}_{[p,1]}(u),
\end{equation*}
with $u\in [0,1]$, $p\in(0,1)$, $r>1$, and $\lambda\in [0, 1/(2(r-1)(1-p)^{r-2})].$\\

In practice, the assessment of $EGS_{r,p}^{\lambda}$ can be reduced to evaluating its discrete version proposed by Acerbi\cite{1} as a consistent estimator for spectral risk measures:
\begin{equation}\label{}
  \widehat{EGS}_{r,p}^{\lambda}(X)=\sum_{i=1}^{N} X_{(i)}\,\phi_{i},
\end{equation}
where, $\{X_{(i)}; i=1,...,N\}$ are the ordered statistics given by the N-tuple $\{X_{1},...,X_{N}\}$ of observations, and $\phi_{i}$ is the natural choice for a suitable weighting function given by:
\begin{equation}\label{}
  \phi_{i}= \frac{\phi_{r,p}^{\lambda}(i/N)}{\sum_{k=1}^{N} \phi_{r,p}^{\lambda}(k/N)}\qquad i=1,...,N
\end{equation}
and satisfying $\sum_{i}\phi_{i}=1.$ \\

According to the expression of $\phi_{r,p}^{\lambda}$, the $EGS_{r,p}^{\lambda}$ is concerned only with losses beyond the $VaR_{p}$, thus we are interested in tail risks.
\begin{rmk}
  When $r$ is sufficiently high, $(1-p)^{r-1}-r(1-u)^{r-1} \rightarrow 0$ thus $\phi_{r,p}^{\lambda} \rightarrow \displaystyle \frac{1}{1-p}$.\\
   That is to say, under this condition, $ EGS_{r,p}^{\lambda}$ is confounded with $ES_{p}$.
\end{rmk}
This remark ensures that, for a highly risk averse investor, we can simply use $ES_{p}$. In other words, $ES_{p}$ can be considered as a limit of $EGS_{r,p}^{\lambda}$ for a quite high risk aversion degree. Furthermore, since more risk averse the investor is less risk he takes then we can claim that $EGS_{r,p}^{\lambda}$ is at least equal to $ES_{p}$ i.e., $EGS_{r,p}^{\lambda}\geq ES_{p}$.\\

In the following, we illustrate the above approach with the use of a numerical example. Our dataset consists of daily returns from the MASI index covering the period of 15th November 2016 to the 15th November 2017, which includes a total of $N=250$ observations. The MASI (Moroccan All Shares Index) is a stock index that tracks the performance of all securities listed in the Casablanca Stock Exchange located at Casablanca in Morocco.\footnote{http://www.casablanca-bourse.com}

The proposed methodology does not make any assumptions about the distribution that describes the data, except that an Augmented Unit Root test is performed to make sure that our data series is stationary. Moreover, as shown in Figure \ref{Fig:Return Graph}, the return graph validates this verification since the series fluctuates around 0 and has no trend.
\begin{figure}[h]
  \centering
  \includegraphics[width=15cm]{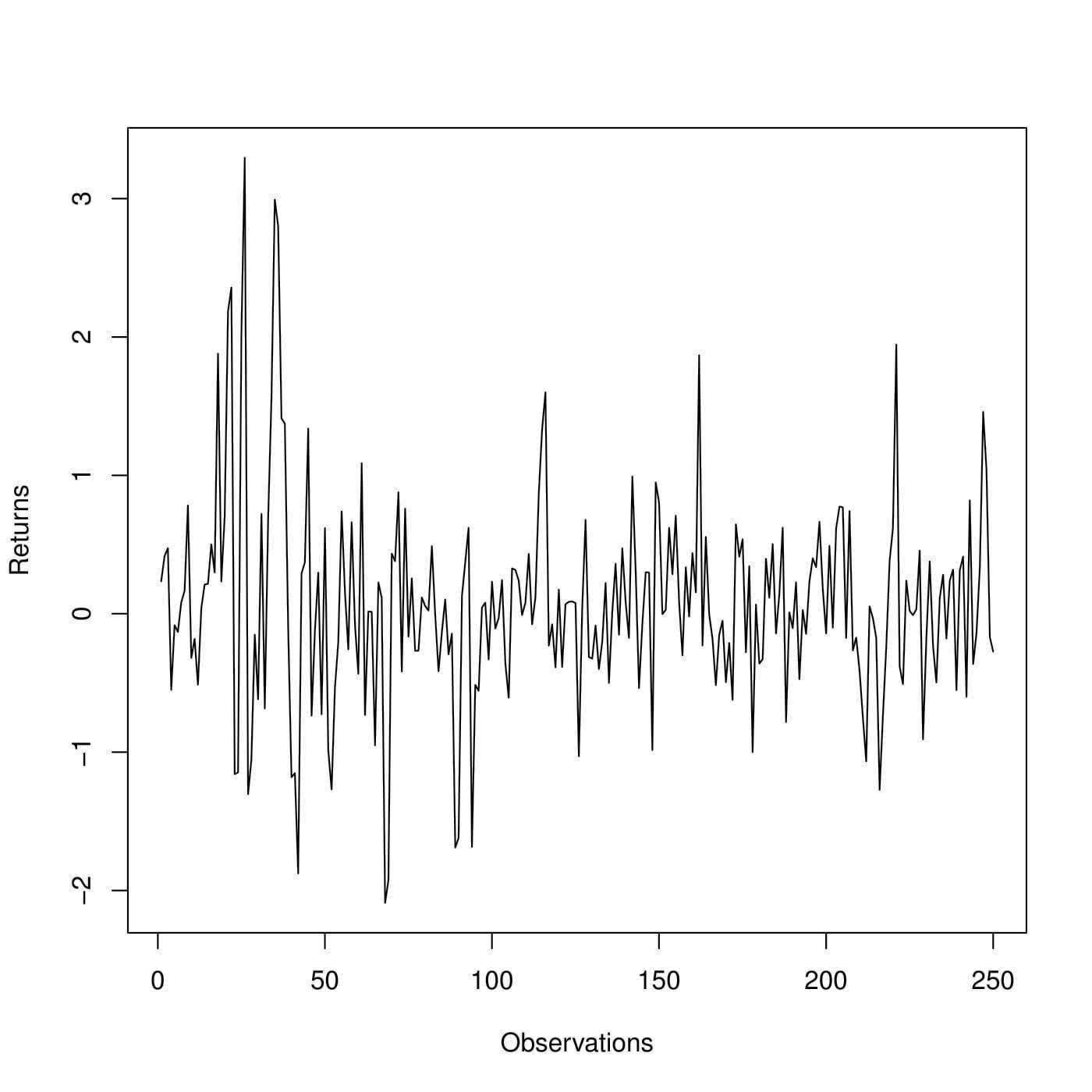}
  \caption{Graph of the daily observed MASI return.}\label{Fig:Return Graph}
\end{figure}

\begin{rmk}
  When defining the $EGS_{r,p}^{\lambda}$ we were making the convention that our rv $X$ represents financial losses (profits) when it has positive (negative) values. However, to be in compliance with the real world data the latter convention is updated.
\end{rmk}
By using the sorted returns series we calculate the $VaR_{p}$ value to identify the concerned losses. Then, we affect to each value its own weight in accordance with the parameters $p, r,$ and $\lambda$ as shown below in Table \ref{Tab: weighted losses}.\\
To fulfill the condition $\lambda\in [0, 1/(2(r-1)(1-p)^{r-2})]$ we take $\lambda$ arbitrarily as the midpoint of the interval, thus $\displaystyle \lambda=\frac{1}{4(r-1)(1-p)^{r-2}}$.\\
\newpage

\begin{table}[t]
\begin{center}
\begin{tabular}{|c|c|}
  \hline
  Sorted Returns(\%) & Weights  \\\hline\hline
      \vdots   &    \vdots  \\
  -1.146 & 0.04  \\
  -1.151 & 0.05  \\
  -1.160 & 0.05  \\
  -1.181 & 0.06 \\
  -1.270 & 0.06  \\
  -1.273 & 0.07  \\
  -1.304 & 0.08  \\
  -1.621 & 0.08 \\
  -1.685 & 0.09  \\
  -1.690 & 0.10 \\
  -1.880 & 0.10 \\
  -1.924 & 0.11  \\
  -2.090 & 0.11  \\
  \hline
\end{tabular}
\end{center}
\caption{Weighted losses beyond VaR for p=95\% and r=2.}\label{Tab: weighted losses}
\end{table}

In Table \ref{Tab: results} below, we report the calculation results of $\widehat{EGS}_{r,p}^{\lambda}$ for different values of $p$ and $r$:
\begin{table}[!htb]
\begin{center}
\begin{tabular}{c|c|c|c|c|c}
  \hline\hline
    $\widehat{EGS}_{r,p}^{\lambda}$ & $r=2 \,(GS^{\lambda}_{p})$ & $r=3$ & $r=6$ & $r=20$ & $r=30$ \\\hline
    $p=90\%$  &   &  &  &  & \\
    &   &  &  &  & \\
    $VaR=0.73\%$ & 1.32\%  & 1.28\% & 1.25\% & 1.22\% & 1.22\% \\
    &   &  &  &  & \\
    $ES=1.21\%$ &   &  &  &  &   \\\hline
    $p=95\%$  &   &  &  &  & \\
    &   &  &  &  & \\
    $VaR=1.11\%$ & 1.58\%  & 1.55\% & 1.52\% & 1.50\% & 1.49\% \\
    &   &  &  &  & \\
    $ES=1.49\%$ &   &  &  &  &   \\\hline
    $p=99\%$ &   &  &  &  &   \\
    &   &  &  &  & \\
    $VaR=1.79\%$ & 1.99\%  & 1.98\% & 1.97\% & 1.96\% &  1.96\% \\
    &   &  &  &  & \\
    $ES=1.96\%$ &   &  &  &  &   \\\hline
\hline
\end{tabular}
\end{center}
\caption{Outcomes of the $EGS$ estimator according to $p$ and $r$.}\label{Tab: results}
\end{table}

This empirical exercise using the daily returns for the MASI index between 15th November 2016 and 15th November 2017 is a historical approach that illustrates the practical use of $EGS_{r,p}^{\lambda}$ in the real world by considering the psychological attitude of the investor. The obtained results confirm earlier remarks in the previous subsection: first, more the investor is risk averse less risks he takes and then smaller is the amount of capital required to hedge his position; furthermore, we have $EGS_{r,p}^{\lambda}\geq ES_{p}$ and in a highly risk averse context ($r\geq 20$ for this survey) both risk measures can be confounded.

\makeatletter

\makeatother

\end{document}